\documentclass[12pt,reqno]{amsart}

\usepackage[utf8]{inputenc}

\usepackage[left = 0.8in, right = 0.8in, top = 0.8in, bottom = 0.8in]{geometry}
\usepackage{amsmath,amsthm,amssymb}
\usepackage{bbm}
\usepackage{graphicx}
\usepackage{color}
\usepackage[sorting=none, url = false, doi = false, eprint = false, isbn = false]{biblatex}
\usepackage{cancel}
\usepackage{todonotes}
\usepackage{cleveref}
\usepackage[normalem]{ulem}

\AddToHook{cmd/appendix/before}{\def\cref@section@alias{appendix}}

\newtheorem{defn}{Definition}
\newtheorem{theorem}{Theorem}
\newtheorem{prop}{Proposition}
\newtheorem{lemma}{Lemma}

\newcommand{\R}{\mathbb{R}}
\newcommand{\N}{\mathbb{N}}

\newcommand{\C}{\mathbb{C}}

\newcommand{\mS}{\mathcal{S}}

\newcommand{\mD}{\mathcal{D}}
\newcommand{\mF}{\mathcal{F}}
\newcommand{\sgn}{\operatorname{sgn}}
\newcommand{\harm}{\frac{1}{2}k}
\newcommand{\fpsi}{\widehat{\phi}}

\newcommand{\la}{\left\langle}
\newcommand{\ra}{\right\rangle}

\DeclareMathOperator{\Airya}{Ai}
\DeclareMathOperator{\Airyb}{Bi}
\usepackage[dvipsnames]{xcolor}

\newcommand{\stkout}[1]{\ifmmode\text{\sout{\ensuremath{#1}}}\else\sout{#1}\fi}

\title{Spectral Theory of an Exactly Solvable Quantum-Optical Oscillator}

\author[]{Sam Johar}
\address{Department of Mathematics, Amherst College, Amherst, MA 01002, USA}
\email{sjohar27@amherst.edu}

\author[]{Joseph Kraisler}
\address{Department of Mathematics, Williams College, Williamstown, MA 01267, USA}
\email{jk34@williams.edu}

\date{\today}

\begin{document}

\begin{abstract}
We study a one-dimensional single-excitation model of a scalar photon field coupled to a continuum of two-level atoms with quadratic density. The resulting matrix Hamiltonian has an unbounded spatially varying coupling and is not semibounded. We prove that it is essentially self-adjoint and determine its spectrum explicitly. The continuous spectrum is the half-line $(-\infty,\Omega]$, where $\Omega$ is the atomic resonance frequency, while the spectrum above $\Omega$ consists of a strictly increasing sequence of simple eigenvalues. These eigenvalues are characterized by the interlaced zeros of the Airy function and its derivative and satisfy
\begin{align*}
    \omega_n = \left(\frac{3\pi g n}{4}\right)^{1/2}  + \frac{\Omega}{4} + O(n^{-1/2}).
\end{align*}
The analysis is based on a nonlinear spectral reduction that yields a confining massless relativistic oscillator above $\Omega$ and an inverted oscillator below $\Omega$.
\end{abstract}

\maketitle

\section{Introduction}

Quantum optical systems provide a framework for studying the interaction of light with matter at small scales. Such models exhibit several fundamental phenomena, including spontaneous emission, resonant energy exchange, radiative transport, and the formation of hybrid light-matter excitations. Due to the increasing application in quantum computing, nanophotonics, and precision measurement, there is a growing interest in mathematical models in quantum optics which can describe interactions with a large number of atoms. One such model describing the interaction between a scalar electromagnetic field and a collection of stationary two-level atoms distributed according to a number density $\rho(x)$ was introduced in \cite{kraisler_22} and subsequently studied in 
\cite{kraisler_23,Hoskins_21,Hoskins_23,NonlocPDE1,NonlocPDE2}. The equations of motion associated to states with at most one excitation are 
\begin{subequations}\label{eq:dynamics}
\begin{align}
      i\partial_t\phi &= c|D|\phi + g\sqrt{\rho(x)}\alpha, \\
    i\partial_t\alpha &= g\sqrt{\rho(x)}\phi + \Omega\alpha,
\end{align}
\end{subequations}
where $c$ is the speed of light in the vacuum, $g>0$ is a positive coupling constant, and $\Omega>0$ is the resonant frequency of the atoms (i.e. $\Delta E=\hbar\Omega$ is the energy difference between the ground and excited states of each atom). The operator $|D|=(-\partial_x^2)^{1/2}$ is a one-dimensional fractional Laplacian which is nonlocal and defined through the Fourier transform
\begin{align}\label{eq:frac_laplacian}
    \widehat{|D|f}(k) = |k|\widehat{f}(k).
\end{align}
Fractional analogs of the Schr\"odinger equation have been introduced in quantum mechanics \cite{Laskin_00,Laskin_2002}, and are well studied in the physics and mathematics literature \cite{zhang_2021,Frank_2008,Kaleta_2010,Felmer_2012, frank_2017,Hiroshima_2012,Daubechies_1983,Carmona_1990}. In our case, the nonlocal operator \eqref{eq:frac_laplacian} arises in a very physically natural way: the dispersion relation for a free electromagnetic field is linear in the momentum, $k\in\R$, and given by the formula $\omega(k)=c|k|$. Applying the inverse Fourier transform to $c|k|$ yields the operator $c|D|$. For more details of this derivation, see \cite{kraisler_22}. For a discussion of several equivalent definitions of this operator which make its nonlocal nature clearer, see \cite{Kwasnicki_2017}. \\

The functions $\phi(x,t)$ and $\alpha(x,t)$ in the system \eqref{eq:dynamics} are each assumed to be in $L^2(\R)$ and the normalization
\begin{align*}
    \int_{\R} |\phi(x,t)|^2 + |\alpha(x,t)|^2 dx = 1,
\end{align*}
is preserved by the unitary evolution generated by the self-adjoint Hamiltonian constructed below. This identity allows one to interpret $|\phi(x,t)|^2$ and $|\alpha(x,t)|^2$ as probability density functions for locating photons and excited atoms respectively. For notational simplicity, we define $L^2(\R;\C^2)=L^2(\R)\oplus L^2(\R)$ so that the pair $\Psi(x,t)=(\phi(x,t),\alpha(x,t))^{\top}\in L^2(\R;\C^2)$ for each $t\geq 0$. The system above can be recast as a single operator-valued matrix equation
\begin{align*}
    i\partial_t\Psi = H_{\rho}\Psi,\quad  \Psi\in L^2(\R;\C^2),
\end{align*}
where the Hamiltonian $H_{\rho}$ is formally defined 
\begin{align*}
    H_{\rho}=\begin{pmatrix}
    c|D| & g\sqrt{\rho}\\ g\sqrt{\rho} & \Omega
\end{pmatrix}.
\end{align*}

In this paper we are primarily interested in determining the spectrum of $H_{\rho}$ in the case that the density of atoms is modeled as a quadratic function of the spatial variable, i.e. $\rho(x)=\rho_0 x^2$. For notational convenience we will absorb the factor of $\sqrt{\rho_0}$ into the coupling constant $g$ and set the speed of light $c=1$. Our main results show that the spectrum of the relevant Hamiltonian decomposes into a continuous half-line and a sequence of simple eigenvalues above the resonant frequency $\Omega$. 
\begin{theorem}\label{thm:spec_nonlocal}
Let 
\begin{align}\label{eq:nonlocal_hamiltonian}
    H_0 = \begin{pmatrix}
        |D| & g|x|\\ g|x| & \Omega
    \end{pmatrix},\quad \mD(H_0)=\mS(\R;\C^2).
\end{align}
Then $H_0$ is essentially self-adjoint and we denote its unique self-adjoint extension by $H$. The spectrum of $H$ is given by
\begin{align*}
    \sigma_{\rm ess}(H) = \sigma_{c}(H) = (-\infty,\Omega],\quad \sigma_p(H)=\{\omega_n\ | \ n\geq 1\},
\end{align*}
where each $\omega_n$ is simple and
\begin{align*}
    \Omega < \omega_1 < \omega_2 < \cdots< \omega_n\to\infty.
\end{align*}
The corresponding eigenvectors are given, up to normalization, by
\begin{align*}
    \Psi_n(x) = \begin{pmatrix}
        \phi_n(x)\\
        \dfrac{g|x|}{\omega_n-\Omega}\phi_n(x)
    \end{pmatrix},\quad \widehat{\phi_n}(\xi) = \sgn(\xi)^{n+1}\Airya\left(\left(\dfrac{\omega_n-\Omega}{g^2}\right)^\frac{1}{3}(|\xi|-\omega_n)\right),\quad n\geq 1
\end{align*}
Moreover, we have the asymptotic formula
\begin{align*}
    \omega_n = \left(\frac{3\pi g n}{4}\right)^{1/2}  + \frac{\Omega}{4} + O(n^{-1/2}),
\end{align*}
from which we have the eigenvalue counting asymptotic
\begin{align*}
    N(\omega) := |\{ n\in\N \ | \ \omega_n\leq \omega\}| = \frac{4}{3\pi g}\omega^2 -\frac{2\Omega}{3\pi g}\omega +O(1).
\end{align*}
\end{theorem}
The proof of the first claim, that $H_0$ is essentially self-adjoint employs an analog of Nelson's Theorem \cite{Nelson_1972} due to Faris and Lavine \cite{Faris_1974} and can be found in Appendix \ref{app:self_adjoint}.\\ 

Although this choice of $\rho(x)$ is not a globally realistic atomic density, it provides an exactly solvable model in which the effect of an unbounded spatially varying light-matter coupling can be analyzed completely. 
Our model follows a substantial body of work on quantum models with quadratic potentials. The problem of studying the Schr\"odinger equation with quadratic potential $V(x) = \harm x^2$, $k>0$, is referred to as the quantum simple harmonic oscillator (SHO) \cite{Griffiths_Schroeter_2018,Hall_13} and it is a fundamental example of an exactly solvable model in single-particle quantum mechanics. More precisely, the operator 
\begin{align}\label{eq:local_oscillator}
    H_{SHO} = -\partial_x^2 + \harm x^2,
\end{align}
has a family of eigenvalues $E_n = \sqrt{\dfrac{k}{2}}(2n+1)$, $n\geq 0$, which are equally spaced, and the associated family of eigenfunctions forms an orthonormal basis of $L^2(\R)$. The harmonic oscillator is a cornerstone of introductory quantum mechanics courses as it provides a reasonable approximation for the behavior of an electron in the presence of a smooth nondegenerate potential well. \\

A related exactly solvable model is the fractional simple harmonic oscillator (FSHO) \cite{Stephanovich_22,Jeng_10,Kowalski_10,Lorinczi_12}, whose Hamiltonian is given by 
\begin{align}\label{eq:nonlocal_oscillator}
    H_{FSHO} = |D| + \harm x^2.
\end{align}
Once again this operator has a discrete set of eigenvalues $E_n$, $n\geq 0$, for which the associated eigenvectors form an orthonormal basis of $L^2(\R)$. However, the eigenvalues are no longer evenly spaced and are given in terms of roots of the Airy function $\Airya$. We present a more detailed description of the spectrum of this operator in \Cref{thm:frac_oscillator}.\\

The rest of the paper is organized as follows. After briefly introducing some notation and convention in \Cref{sec:Notation}, we construct the self-adjoint Hamiltonian and develop the scalar spectral reduction in \Cref{sec:spec_prob}. In \Cref{sec:Eigenvalues} we complete the proof of \Cref{thm:spec_nonlocal} by determining the spectrum explicitly. Finally we include two appendices: one collecting some spectral properties of the inverted fractional harmonic oscillator whose precise statement the authors have not found in the literature, and a brief discussion of the local analog of the Hamiltonian \eqref{eq:nonlocal_hamiltonian}.

\section{Notation and Conventions}\label{sec:Notation}
We briefly recall some concepts and standardize our notation for the rest of the paper. 
\begin{itemize}
    \item The Fourier transform and its inverse are defined by
    \begin{align*}
        \mF[f](\xi)=\widehat{f}(\xi) = \int_{\R} f(x)e^{-i\xi x}dx,\quad \mF^{-1}[g](x) = \int_{\R} g(\xi)e^{i\xi x}\frac{d\xi}{2\pi}.
    \end{align*}
    \item We use the notation $f\lesssim g$ to mean there exists a constant $C>0$ such that
\begin{align*}
    f(x)\leq Cg(x).
\end{align*}
    \item We write $D = -i\partial_x$ so that the one-dimensional fractional Laplacian can be written $(-\Delta)^{1/2} = |D|$.
\end{itemize}

\section{The Hamiltonian  \eqref{eq:nonlocal_hamiltonian}}\label{sec:spec_prob}

The spectral problem associated to the Hamiltonian
\eqref{eq:nonlocal_hamiltonian} is
\begin{align*}
    H\Psi = \omega\Psi,\quad  \Psi\in L^2(\R;\C^2),
\end{align*}
for some $\omega\in \R$. Equivalently, this may be written as the pair of equations 
\begin{subequations}\label{eq:spectral_problem}
\begin{align}
      |D|\phi + g|x|\alpha&=\omega\phi, \\
    g|x|\phi + \Omega\alpha&=\omega\alpha ,
\end{align}
\end{subequations}
for $\phi,\alpha\in L^2(\R)$. We start by proving the resonant frequency of the atoms, $\Omega$, is not an $L^2(\R;\C^2)$-eigenvalue of $H$. 
\begin{lemma}\label{lem:no_resonance}
There are no nonzero solutions $(\phi,\alpha)^{\top}\in L^2(\R;\C^2)$ to the system \eqref{eq:spectral_problem} for $\omega = \Omega$.
\end{lemma}
\begin{proof}
Replacing $\omega$ by $\Omega$, the second equation becomes $g|x|\phi =0$ which implies $\phi=0$ in $L^2(\R)$. Substituting this back into the first equation we find $g|x|\alpha=0$ which similarly implies $\alpha=0$ in $L^2(\R)$. 
\end{proof}
Thus, if $\omega\in\R$ is an eigenvalue, the second equation may be solved for $\alpha$ and substituted back into the first equation to obtain
\begin{align*}
    \omega\phi = |D|\phi + \frac{g^2x^2}{\omega-\Omega}\phi.
\end{align*}
This motivates the following definition.
\begin{defn}
For $\omega\neq\Omega$, the operator $L_{\omega}$ is defined as the self-adjoint closure of
\begin{align*}
    L_{\omega}\phi = |D|\phi + \frac{g^2x^2}{\omega-\Omega}\phi,
\end{align*}
initially defined on $\mS(\R)$. We note that for $\omega>\Omega$, $\mD(L_{\omega}) = \mD(|D|)\cap \mD(x^2)$.
\end{defn}
It is also useful to define the operator $S_{\omega}=L_{\omega}-\omega$. We have the following lemma, which relates the spectrum of $H$ to that of $S_{\omega}$. 
\begin{lemma}\label{lem:spectral_criterion}
For $\omega\neq\Omega$, let $L_{\omega}$ and $S_{\omega}$ be defined as above. Then:
\begin{enumerate}
    \item $\omega\in \sigma_{p}(H)$ if and only if $\omega\in \sigma_{p}(L_{\omega})$. In particular
    \begin{align*}
        \ker(H-\omega) = \left\{ \begin{pmatrix}
            \phi \\ \dfrac{g|x|\phi}{\omega-\Omega}\end{pmatrix} \ | \ \phi\in \ker S_{\omega}
        \right\}.
    \end{align*}
    \item For $\omega > \Omega $ we have $\omega\in \rho(H)$ if and only if $0\in\rho(S_{\omega})$. In this case the inverse is given by
    \begin{align*}
        (H-\omega)^{-1} = \begin{pmatrix}
            S_{\omega}^{-1} & \frac{g}{\omega-\Omega}S_{\omega}^{-1}|x|\\
            \frac{g}{\omega-\Omega}|x|S_{\omega}^{-1} & -\frac{1}{\omega-\Omega} + \frac{g^2}{(\omega-\Omega)^2}|x|S_{\omega}^{-1}|x|
        \end{pmatrix}.
    \end{align*}
\end{enumerate}
\end{lemma}
Combining this result with \Cref{lem:no_resonance}, we see that the $L^2(\R)$ solutions to \eqref{eq:spectral_problem} are completely determined by the eigenvalues of the one-parameter family of operators $\{ L_{\omega}\}_{\omega\neq \Omega}$.
\begin{proof}
Fix $\omega\neq\Omega$ and suppose that $\Psi=(\phi,\alpha)\in L^2(\R;\C^2)$ is an eigenvector of $H$ with eigenvalue $\omega$. The second component of the equation $(H-\omega)\Psi=0$ gives
\begin{align*}
    g|x|\phi -(\omega-\Omega)\alpha = 0 \implies \alpha = \frac{g|x|\phi}{\omega-\Omega}.
\end{align*}
As $|D|\phi+g|x|\alpha \in L^2$, we also have $\phi\in \mD(L_{\omega})$. Substituting the equation for $\alpha$ into the first equation yields
\begin{align*}
     (|D| + \frac{g^2x^2}{\omega-\Omega}-\omega)\phi = 0.
\end{align*}
 In particular, we find $L_{\omega}\phi = \omega\phi$ or $S_{\omega}\phi = 0$. Conversely, suppose that $\phi\in \ker(S_{\omega})$. Then we can define 
\begin{align*}
    \alpha = \frac{g|x|\phi}{\omega-\Omega}.
\end{align*}
As $\phi\in\mD(L_{\omega})$ we have $|x|\phi,\alpha\in L^2$. Additionally,
\begin{align*}
    |D|\phi + g|x|\alpha = |D|\phi + \frac{g^2 x^2}{\omega-\Omega}\phi = L_{\omega}\phi = \omega\phi \in L^2.
\end{align*}
Thus $(\phi,\alpha)^{\top}\in \mD(H)$ and reversing the calculation above shows
\begin{align*}
    H\begin{pmatrix}
        \phi\\ \alpha
    \end{pmatrix} = \omega\begin{pmatrix}
        \phi\\  \alpha
    \end{pmatrix}.
\end{align*}

Towards the proof of the second statement, suppose $\omega > \Omega$. If $0\in \sigma(S_{\omega})$ then as $S_{\omega}$ has compact resolvent, its spectrum is composed of discrete eigenvalues and hence $0\in\sigma_{p}(S_{\omega})$. By the first part of the lemma, $\omega\in \sigma_p(H)$ and hence $\omega\in \sigma(H)$. Conversely, suppose that $0\in\rho(S_{\omega})$. Then $S_{\omega}^{-1}$ exists as a bounded operator on $L^2$ and formally we may write the following expression for $(H-\omega)^{-1}$ using the Schur complement:
    \begin{align*}
        (H-\omega)^{-1} = \begin{pmatrix}
            S_{\omega}^{-1} & \frac{g}{\omega-\Omega}S_{\omega}^{-1}|x|\\
            \frac{g}{\omega-\Omega}|x|S_{\omega}^{-1} & -\frac{1}{\omega-\Omega} + \frac{g^2}{(\omega-\Omega)^2}|x|S_{\omega}^{-1}|x|
        \end{pmatrix}.
    \end{align*}
It remains to check that each entry defines a bounded operator on $L^2$.  As $L_{\omega}$ is a strictly positive, invertible operator satisfying 
\begin{align*}
    x^2 \leq \frac{(\omega-\Omega)}{g^2} L_{\omega},
\end{align*}
in the sense of quadratic forms, $|x|L_{\omega}^{-1/2}$ extends to a bounded operator. Moreover, $L_{\omega}^{1/2}S_{\omega}^{-1}$ is bounded as corresponds to the bounded function
\begin{align*}
    f(\lambda) = \frac{\sqrt{\lambda}}{\lambda-\omega},\quad \lambda \in \sigma(L_{\omega}). 
\end{align*}
through the spectral calculus for $L_{\omega}$. Similarly, the operator $L_{\omega}^{1/2}S_{\omega}^{-1}L_{\omega}^{1/2}$ is bounded as it corresponds to $g(\lambda)=\dfrac{\lambda}{\lambda-\omega}$. This is sufficient as now
\begin{align*}
    |x|S_{\omega}^{-1} = (|x|L_{\omega}^{-1/2})(L_{\omega}^{1/2}S_{\omega}^{-1}),
\end{align*}
is the composition of bounded operators and hence bounded. Additionally, $S_{\omega}^{-1}|x|=(|x|S_{\omega}^{-1})*$ and 
\begin{align*}
    |x|S_{\omega}^{-1}|x| = (|x|L_{\omega}^{-1/2})(L_{\omega}^{1/2}S_{\omega}^{-1}L_{\omega}^{1/2})(L_{\omega}^{-1/2}|x|),
\end{align*}
are both bounded as well.

\end{proof}

For each fixed $\omega\neq\Omega$, the operator $L_{\omega}$ is a one-dimensional Schr\"odinger operator with a quadratic potential. There are two cases.
\begin{itemize}
    \item For $\omega > \Omega$, $L_{\omega}$ is a fractional simple harmonic oscillator (FSHO).
    \item For $\omega < \Omega$, $L_{\omega}$ is a fractional inverted harmonic oscillator (FIHO).
\end{itemize}
As mentioned in the introduction, FSHOs have been well studied in physics and mathematics literature. We summarize the results below, a detailed proof of which can be found in \cite[Section 3]{Lorinczi_12}
\begin{theorem}[FSHO]\label{thm:frac_oscillator}
Let $H_{FSHO} = |D| + \harm x^2$. Then $H_{FSHO}$ is a positive self-adjoint operator on $L^2(\R)$ with purely discrete and simple spectrum $\sigma(H_{FSHO})=\{\lambda_n\}_{n\geq 1}$.
\begin{align*}
    \lambda_n = \left(\harm\right)^{1/3}r_n,
\end{align*}
where $r_{n}$ are the interlaced roots of the Airy function $\Airya(-z)$, for even $n$, and its derivative $\Airya'(-z)$, for odd $n$.
\end{theorem}

Although the inverted harmonic oscillator (IHO) $-\partial_x^2+\harm x^2$ has been wellstudied \cite{Barton_87,Yuce_2006,Bolte_2017,Finster_2017}, FIHOs have not received much attention. For this reason, we summarize their spectral properties in a theorem which we prove in Appendix \ref{app:frac_inverted}.

\begin{theorem}[FIHO]\label{thm:inverted_oscillator}
The operator $H_{FIHO}=|D|-\harm x^2$ is a self-adjoint operator on $L^2(\R)$. $H_{FIHO}$ has no $L^2(\R)$-eigenvalues and the spectrum is the entire real line $\sigma(H_{FIHO})=\R$.
\end{theorem}

\section{Spectrum of the Hamiltonian \eqref{eq:nonlocal_hamiltonian} }\label{sec:Eigenvalues}
In this section we study the $L^2(\R;\C^2)$-eigenvalues of the operator $H$ in \eqref{eq:nonlocal_hamiltonian} and ultimately prove our main result \Cref{thm:spec_nonlocal}. We first prove that there cannot exist any eigenvalues of $H$ which are smaller than the resonant frequency $\Omega$.
\begin{lemma}\label{lem:no_eigs}
If $\omega\in\R$ is an eigenvalue of $H$, then $\omega > \Omega$.     
\end{lemma}

\begin{proof}
By \Cref{lem:no_resonance}, we already know that $\omega\neq \Omega$. Now suppose that $\omega < \Omega$ is an eigenvalue of $H$. By \Cref{lem:spectral_criterion}, this means that $\omega$ is an eigenvalue of $L_{\omega}$. However, making the identification
\begin{align*}
    \harm = \frac{g^2}{\Omega-\omega}, 
\end{align*}
we see that $L_{\omega} = |D|-\harm x^2$ cannot have any eigenvalues by \Cref{thm:inverted_oscillator}, contradicting the assumption that $\omega<\Omega$. 
\end{proof}

\begin{proof}[Proof of \Cref{thm:spec_nonlocal}]
By Proposition \ref{prop:self_adjoint} in Appendix \ref{app:self_adjoint}, it remains to determine the spectrum of $H$. Fix $\omega>\Omega$. \Cref{thm:frac_oscillator} states that the point spectrum of the operator $L_{\omega}$ is well understood for fixed $\omega$, but our goal is to determine which $\omega$ are eigenvalues of the corresponding operator $L_{\omega}$. We follow an argument similar to that of \cite[Section 3]{Lorinczi_12} along with an application of the intermediate value theorem. To this end, consider the eigenvalue problem
\begin{align*}
    L_{\omega}\phi =\lambda(\omega)\phi,\quad \phi\in L^2(\R),
\end{align*}
where we note that the eigenvalues depend on $\omega$. For the time being, we will suppress this dependence. Applying the Fourier transform we have
\begin{align*}
    |\xi|\fpsi - \frac{g^2}{\omega-\Omega}\partial_{\xi}^2\fpsi = \lambda\fpsi\implies \partial_{\xi}^2\fpsi - \frac{(\omega-\Omega)(|\xi|-\lambda)}{g^2}\fpsi=0,
\end{align*}
which is a Schr\"odinger operator with a piecewise linear trapping potential. As the potential is an even function, every solution can be obtained from even and odd solutions. For $\xi>0$ we make the change of variables $z = \left(\dfrac{\omega-\Omega}{g^2}\right)^\frac{1}{3}(\xi-\lambda)$ and obtain the Airy equation
\begin{align}\label{eq:Airy}
    \partial_z^2\fpsi -z\fpsi =0 ,
\end{align}
where we have kept the notation $\fpsi$ by an abuse of notation. Every solution to \eqref{eq:Airy} is a linear combination of $\Airya(z)$ and $\Airyb(z)$. As $\displaystyle\lim_{z\to\infty} \Airyb(z)=\infty$, we must have that
\begin{align*}
    \fpsi(\xi) = \Airya\left(\left(\dfrac{\omega-\Omega}{g^2}\right)^\frac{1}{3}(\xi-\lambda)\right),\quad \xi > 0.
\end{align*}
The odd and even solutions are, respectively, characterized by the conditions at $\xi=0$:
\begin{align*}
     \Airya\left(-\left(\dfrac{\omega-\Omega}{g^2}\right)^\frac{1}{3}\lambda)\right)=0,\quad  \Airya'\left(-\left(\dfrac{\omega-\Omega}{g^2}\right)^\frac{1}{3}\lambda)\right)=0.
\end{align*}
The functions $\Airya(z)$ and $\Airya'(z)$ each have a discrete set of zeros on the negative real axis, which we call $\{a_k\}$ and $\{a_k'\}$ respectively. We use $r_n$ for $n\geq1$ to denote the sequence of (interlaced) roots of $\Airya$ and $\Airya'$. Specifically, 
\begin{align*}
    r_{2k-1} = |a_k'|,\quad r_{2k} = |a_k|,\quad k\geq 1.
\end{align*}
Note that $r_n$ is an increasing sequence of positive real numbers. For each $\omega>\Omega$, this provides a sequence of simple eigenvalues $\{\lambda_n(\omega)\}$ defined
\begin{align*}
    \lambda_n(\omega) = \left(\frac{g^2}{\omega-\Omega}\right)^{1/3}r_n.
\end{align*}
To recover the eigenvalues of $H$, we seek solutions to the identity $\lambda_n(\omega)=\omega$. For each $n\in\N$, the function $\lambda_n:(\Omega,\infty)\to (0,\infty)$ has the following properties
\begin{enumerate}
    \item $\lambda_n$ is a differentiable function of $\omega$.
    \item $\lambda_n$ is a decreasing function of $\omega$. 
    \item $\lambda_n$ has the following asymptotics
    \begin{align*}
        \lim_{\omega\to\Omega^{+}}\lambda_n(\omega) = +\infty,\quad \lim_{\omega\to\infty}\lambda_n(\omega) = 0.
    \end{align*}
\end{enumerate}
The first and third properties are immediate. The second follows as 
\begin{align*}
    \frac{d\lambda_n}{d\omega} = -r_n\frac{g^{2/3}}{3}\left(\omega-\Omega\right)^{-4/3} < 0, \quad \omega > \Omega.
\end{align*}
Hence, by the intermediate value theorem, the graph of $\lambda_n(\omega)$ crosses the diagonal exactly once, see \Cref{fig:crossings}, for each $n$; call this value $\omega_n$. As $r_n$ is strictly increasing, it follows that $\omega_n$ is strictly increasing as well. The corresponding eigenfunctions are, up to normalization, the inverse Fourier transforms of
\begin{align*}
    \fpsi_n(\xi) &= \sgn(\xi)^{n+1}\Airya\left(\left(\dfrac{\omega_n-\Omega}{g^2}\right)^\frac{1}{3}(|\xi|-\omega_n)\right),\quad n\geq 1.
\end{align*}
Finally, to obtain the asymptotics, we see that each $\omega_n$ must satisfy
\begin{align*}
     \omega_n = \left(\frac{g^2}{\omega_n-\Omega}\right)^{1/3}r_n\implies \omega_n^4 -\Omega\omega_n^3 -g^2r_n^3 = 0.
\end{align*}
There is precisely one positive real root to the fourth degree polynomial above. It is clear that as $r_n\to\infty$, $\omega_n\to\infty$. Known asymptotics of the Airy function $\Airya(z)$ and its derivative \cite{Stegun_64} provide the leading order asymptotics for the sequence $r_n$.
\begin{align*}
    r_n = \left(\frac{3\pi n}{4}\right)^{2/3} + O(n^{-1/3}).
\end{align*}
Thus for sufficiently large $n$
\begin{align*}
    \omega_n = \left(\frac{3\pi g n}{4}\right)^{1/2}  + \frac{\Omega}{4} + O(n^{-1/2}).
\end{align*}
Letting $a = \left(\frac{3\pi g}{4}\right)^{1/2}$, the asymptotics $\omega_n = a\sqrt{n}+\frac{\Omega}{4}+O(n^{-1/2})$ can be inverted to obtain
\begin{align*}
    n = \frac{1}{a^2}\left(\omega_n - \frac{\Omega}{4}\right)^2 + O(1),
\end{align*}
which implies
\begin{align*}
    N(\omega) = \frac{4}{3\pi g}\omega^2 -\frac{2\Omega}{3\pi g}\omega +O(1).
\end{align*}
If $\omega > \Omega$ then $L_{\omega}$ is a positive operator with compact resolvent. Thus $S_{\omega}=L_{\omega}-\omega$ fails to be invertible precisely when $\omega \in \sigma_p(L_{\omega})$, i.e. at one of the $\omega_n$ for $n\geq 1$. If $\omega>\Omega$ is not one of these values, then $0\in\rho(S_{\omega})$ and hence $\omega\in \rho(H)$ by Lemma \ref{lem:spectral_criterion}.\\ 

Next, consider $\omega < \Omega$. Then once again by \Cref{thm:inverted_oscillator}, the operator $L_{\omega}$ has spectrum which is equal to all of $\R$. Thus by Weyl's criterion there is a sequence $\phi_n\in\mD(L_{\omega})$ of nonzero vectors such that
\begin{align*}
    \frac{\|(L_{\omega}-\omega)\phi_n\|_{L^2(\R)}}{\|\phi_n\|_{L^2(\R)}}\to 0.
\end{align*}
As $\mS(\R)$ is a core for $L_{\omega}$, there is no harm in assuming that Weyl sequence is chosen so that $\phi_n\in \mS(\R)$. Define $\alpha_n(x) = \dfrac{g|x|}{\omega-\Omega}\phi_n(x)$. Then we have $|x|\phi_n, |D|\phi_n,|x|\alpha_n\in L^2(\R)$, whence $\Psi_n=(\phi_n,\alpha_n)\in\mD(H)$. Moreover
\begin{align*}
    (H-\omega)\Psi_n = \begin{pmatrix}
        (L_{\omega}-\omega)\phi_n\\ 0
    \end{pmatrix},
\end{align*}
and $\|\Psi_n\|_{L^2(\R;\C^2)}^2 = \displaystyle\int_{\R}\left(1+\frac{g^2 x^2}{(\omega-\Omega)^2}\right)|\phi_n(x)|^2 dx\geq \|\phi_n\|_{L^2(\R)}^2$. This implies
\begin{align*}
    \frac{\|(H-\omega)\Psi_n\|_{L^2(\R;\C^2)}}{\|\Psi_n\|_{L^2(\R;\C^2)}} \leq  \frac{\|(L_{\omega}-\omega)\phi_n\|_{L^2(\R)}}{\|\phi_n\|_{L^2(\R)}}\to 0.
\end{align*}
By \cite[Theorem 9.18]{Hall_13}, we have $\omega\in\sigma(H)$. As the spectrum of $H$ is closed, we also have $\Omega\in\sigma(H)$. Lemmas \ref{lem:no_resonance} and \ref{lem:no_eigs} show that there are no eigenvalues in $(-\infty,\Omega].$ As the eigenvalues $\omega_n$, for $n\geq 1$, are simple and isolated, they belong to the discrete spectrum of $H$. Consequently $\sigma_{\rm ess}(H)=(-\infty,\Omega]$ and $\sigma_{p}(H)=\{\omega_n\}$. 
\end{proof}

\begin{figure}[ht]
    \centering
    \includegraphics[width=\textwidth]{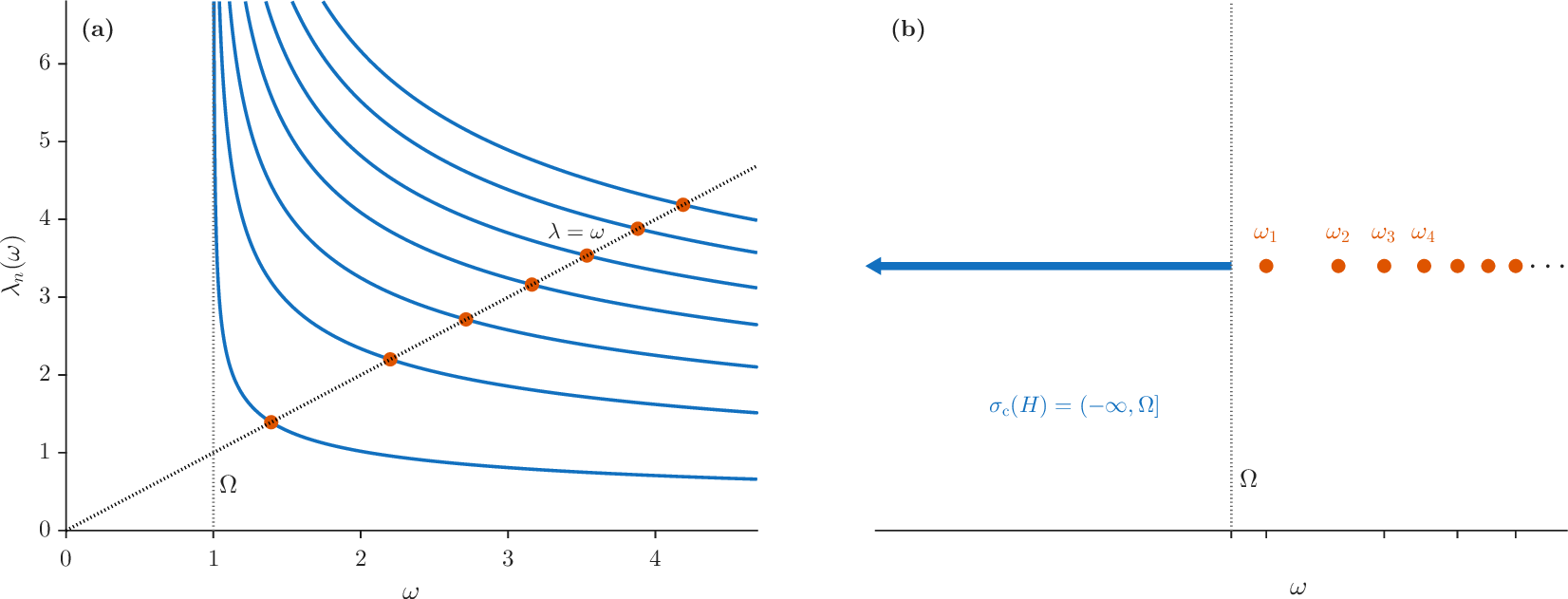}
    \caption{(a) The first seven decreasing eigenvalue branches $$\lambda_n(\omega) = \left(\frac{g^2}{\omega-\Omega}\right)^{1/3}r_n,\quad \omega>\Omega.$$ Each branch intersects the diagonal $\lambda=\omega$ once and these intersections are the eigenvalues $\omega_n$ of the Hamiltonian $H$. (b) A schematic of the spectrum $\sigma(H)$ showing the continuous half-line $(-\infty,\Omega]$ and the first several eigenvalues $\omega_n$. }
    \label{fig:crossings}
\end{figure}

\section{Discussion}

The spectrum of the Hamiltonian \eqref{eq:nonlocal_hamiltonian} arises due to the combination of several physical mechanisms. The field has a linear dispersion relation governed by the fractional Laplacian $|D|$, while the matter has a flat band at $\Omega$. The coupling term $g|x|$ mixes these two components, most strongly at large $x$ values. For any $\omega\neq\Omega$, after eliminating the atomic component, the spectral problem reduces to studying the nonlinear eigenvalue problem
\begin{align*}
    |D|\phi + \frac{g^2 x^2}{\omega-\Omega}\phi = \omega\phi.
\end{align*}
For $\omega > \Omega$ this equation has a discrete set of solutions $\omega_n$ converging to infinity such that $\omega_n = O(n^{1/2})$. For $\omega < \Omega$, this equation has no solutions in $L^2(\R)$, but has a sequence of normalized approximate solutions in the sense that
\begin{align*}
    \left\| |D|\phi_n + \frac{g^2 x^2}{\omega-\Omega}\phi_n-\omega\phi_n\right\|_{L^2(\R)}\to 0.
\end{align*}
As the spectrum of $H$ is a closed set, we also have $\Omega\in\sigma(H)$. Thus we find that the spectrum of the matrix-valued Hamiltonian is given by
\begin{align*}
    \sigma(H)=(-\infty,\Omega]\cup \{\omega_n\},
\end{align*}
with a spectral gap $(\Omega,\omega_1)$ of positive width. This is an example of an exactly solvable nonlocal model in quantum optics. The nonlocality naturally arises as the real-space representation of a field with linear dispersion relation. Moreover, the nonlocality of the operator has a clear impact on the growth rate of the eigenvalues. For the operator \eqref{eq:nonlocal_hamiltonian} the eigenvalues grow asymptotically as
\begin{align*}
    \omega_n \sim \left(\frac{3\pi g}{4}\right)^{1/2}n^{1/2},\quad n\to\infty. 
\end{align*}
while the analysis of the local operator in Appendix \ref{app:local} gives eigenvalues which behave as $n^{2/3}$ for large $n$. \\

Although the quadratic density is highly idealized, as it produces an unbounded coupling to an infinite atomic medium, the model shows that large atomic densities can provide hybrid light-matter bound states above the resonant frequency $\Omega$ with a continuous spectral band below. 

\section*{Acknowledgements}
This work was supported in part by Amherst College through the Gregory S. Call Academic Intern program, which provided undergraduate research funding for SJ.

\appendix

\section{Fractional Inverted 
Oscillator}\label{app:frac_inverted}
In this section we prove \Cref{thm:inverted_oscillator} on the spectral properties of the fractional inverted oscillator
\begin{align*}
    H = |D| - \harm x^2.
\end{align*}
First we cite a result due to Schnol \cite{schnol_1957,Shubin_1992}.
\begin{theorem}[Schnol]\label{thm:Schnol}
Let $H = -\Delta + V$, $V\in L^{\infty}_{loc}(\R^d)$, be a Schr\"odinger operator on $\R^d$ with $V_{-}(x) = -\min(0,V(x))=o(|x|^2)$. If there exists a solution to
\begin{align*}
    Hu = \lambda u,
\end{align*}
with $|u(x)|\leq C(1+|x|)^N$, for some constant $C>0$, then $\lambda\in \sigma(H)$.
\end{theorem}

\begin{proof}[Proof of \Cref{thm:inverted_oscillator}]

Taking the Fourier transform we have
\begin{align*}
    \mathcal{F}H\mathcal{F}^{-1} = |\xi| +\harm\partial_{\xi}^2 = -\harm\left( -\partial_{\xi}^2- a|\xi|\right),\quad a = \left(\harm\right)^{-1}.
\end{align*}
Up to the overall factor of $-\harm$, this is the Hamiltonian $\widetilde{H}= -\partial_{\xi}^2-a|\xi|$, which is essentially self-adjoint on $C_0^{\infty}(\R)$ \cite[Theorem X.38]{simon_reed_V2}. Moreover, one observes that any solution to the eigenvalue problem
\begin{align*}
    \widetilde{H}\phi = \lambda\phi,
\end{align*}
must satisfy $-\phi''-a|\xi|\phi = \lambda\phi$. Considering $\xi>0$ and making the change of variables $z=z_0+a^{1/3}\xi$ with $z_0=a^{-2/3}\lambda$ this reduces to the Airy equation
\begin{align*}
    \partial_z^2\phi + z\phi = 0\implies \phi(z) = c_1\Airya(-z) + c_2\Airyb(-z).
\end{align*}
The functions $\Airya(-z)$ and $\Airyb(-z)$ satisfy the following asymptotics as $z\to\infty$ \cite{Stegun_64}
\begin{align*}
    \Airya(-z) &\sim \frac{1}{\sqrt{\pi}}\frac{1}{z^{1/4}}\sin\left(\frac{2}{3}z^{3/2}+\frac{\pi}{4}\right), \\
    \Airyb(-z) &\sim \frac{1}{\sqrt{\pi}}\frac{1}{z^{1/4}}\cos\left(\frac{2}{3}z^{3/2}+\frac{\pi}{4}\right).
\end{align*}
Due to these asymptotics, any nonzero choice of $c_1$ and $c_2$ leads to a function which is not square integrable at $+\infty$. Hence there are no eigenvalues of $\widetilde{H}$ and correspondingly no eigenvalues of $H$.\\

Now we prove, using Theorem \ref{thm:Schnol}, that every $\lambda\in\R$ is in the spectrum, $\sigma(\widetilde{H})$, of $\widetilde{H}$. The necessary hypothesis on the potential are satisfied as $V_{-}(\xi)=a|\xi|=o(|\xi|^2)$. We proceed by constructing an even solution to the eigenvalue problem above.  Fix $\lambda\in\R$ and for $\xi\geq 0$ define
\begin{align*}
    u_{+}(\xi) = \Airyb'(-z_0)\Airya(-z)-\Airya'(-z_0)\Airyb(-z),\quad z = z_0+a^{1/3}\xi,\quad z_0 =  a^{-2/3}\lambda.
\end{align*}
We see that this function is nonzero as
\begin{align*}
    u_{+}(0) = \Airyb'(-z_0)\Airya(-z_0)-\Airya'(-z_0)\Airyb(-z_0) = \frac{1}{\pi}.
\end{align*}
Finally let $u_{\lambda}(\xi)=u_{+}(|\xi|)$. Clearly $u_{\lambda}$ is even and satisfies $(\widetilde{H}-\lambda)u_{\lambda} = 0$. 
Moreover, $u_{\lambda}'(0) = 0$ and hence $u_{\lambda}$ extends to a $C^1$ function. To see that $u_{\lambda}$ is in fact, $C^2$, note that it satisfies
\begin{align*}
    u''_{\lambda} = -(a|\xi|+\lambda)u_{\lambda}.
\end{align*}
By the asymptotics of the functions $\Airya(z)$ and $\Airyb(z)$, we have the asymptotic bound
\begin{align*}
    |u_{\lambda}(\xi)|^2 & \leq C_{\lambda} z^{-1/2},\quad z\to\infty.
\end{align*}
Hence by \Cref{thm:Schnol}, $\lambda\in\sigma(\widetilde{H})$. Consequently, $-\harm\lambda \in\sigma(H)$, which implies $\sigma(H)=\R$. 
    
\end{proof}

\section{Essential self-adjointness of the Hamiltonian $H_0$}\label{app:self_adjoint}
In this Appendix we prove the Hamiltonian \eqref{eq:nonlocal_hamiltonian} is essentially self-adjoint.

\begin{prop}\label{prop:self_adjoint}
The operator $H_0 = \begin{pmatrix}
        |D| & g|x|\\ g|x| & \Omega
    \end{pmatrix}$ defined on $\mS(\R;\C^2)$ is essentially self-adjoint. 
\end{prop}

\begin{proof}
We apply the Faris-Lavine commutator criterion \cite[Corollary 1.1]{Faris_1974}. First note that $H_0$ is symmetric on $\mS(\R;\C^2)$ by a formal calculation. Next, define the operators $h_0 = -\partial_x^2+x^2+1$ and $N = \begin{pmatrix}
    h_0 & 0\\ 0 & h_0
\end{pmatrix}$. The operator $N$ is a self-adjoint operator and satisfies $N\geq I$. Additionally, the Schwartz class $\mS(\R;\C^2)$ is a core for $N$ for which $\mS(\R;\C^2)\subset \mD(H_0)$. We verify the two hypotheses of the theorem. For $f\in\mS(\R)$, define the quadratic form associated to $h_0$.
\begin{align*}
    q(f):=\la f, h_0 f\ra  = \|f'\|_{2}^2+ \|xf\|_{2}^2 + \| f\|_{2}^2.
\end{align*}
As $h_0\geq I$, we have $\|f\|_2 \leq \|h_0f\|_{2}.$ Therefore,
\begin{align*}
    q(f) =\la f,h_0f\ra \leq \|f\|_2 \|h_0f\|_2 \leq \|h_0f\|_2^2.
\end{align*}
In particular
\begin{align*}
    \| |x|f\|_2\leq \|h_0f\|_2,\quad \| |D|f\|_2 = \| f'\|_2\leq \|h_0f\|_2.
\end{align*}
For $\Psi = (\phi,\alpha)^{\top}\in\mS(\R;\C^2)$ we have
\begin{align*}
    H_0\Psi = \begin{pmatrix}
        |D|\phi + g|x|\alpha\\
        g|x|\phi + \Omega\alpha
    \end{pmatrix},
\end{align*}
and 
\begin{align*}
    \|H_0\Psi\|_2^2 &= \||D|\phi + g|x|\alpha\|_2^2 + \|g|x|\phi + \Omega\alpha\|_2^2\\
    &\leq 2(\||D|\phi\|_2^2 + g^2\||x|\alpha\|_2^2 + g^2\||x|\phi\|_2^2 + \Omega^2\|\alpha\|_2^2)\\
    &\leq C\left(\|h_0\phi\|_2^2 + \|h_0\alpha\|_2^2\right) = C\|N\Psi\|_2^2.
\end{align*}
Thus it remains to show that there exists $C>0$ such that
\begin{align*}
    |\la \Psi,[H_0,N]\Psi\ra|\leq C\la \Psi,N\Psi\ra,\quad \forall\Psi = (\phi,\alpha)^{\top}\in \mS(\R;\C^2).
\end{align*}
The only terms in the commutator which do not vanish are 
\begin{align*}
    [|x|,-\partial_x^2],\quad  [|D|,x^2].
\end{align*}
We start by estimating the first. For $f\in\mS(\R)$
\begin{align*}
    [|x|,-\partial_x^2]f = 2\delta(x)f(0) + 2\sgn(x)f'(x),
\end{align*}
in the sense of distributions. Additionally, $|f(0)|\leq q(f)^{1/2}$ as
\begin{align*}
    |f(0)|^2 &=\int_{-\infty}^0 \frac{d}{dx}|f(x)|^2 dx\\
    &\leq 2\int_{-\infty}^0 |f(x)||f'(x)| dx\\
    &\leq 2\|f\|_{2}\|f'\|_{2}\\
    &\leq \|f\|_{2}^2+\|f'\|_{2}^2\leq q(f).
\end{align*}
Whence for $f,h\in \mS(\R)$
\begin{align*}
    |\la h, [|x|,-\partial_x^2]f\ra | &\leq 2|f(0)||h(0)| + 2\|f'\|_{2}\|h\|_{2}
    \leq C_1(q(f)^{1/2}q(h)^{1/2}) \leq C_2(q(f)+q(h)).
\end{align*}
The commutator $[|D|,x^2]$ satisfies the same estimate (up to a factor of $2\pi$). Indeed, by the properties of the Fourier transform
\begin{align*}
    \mF[|D|,x^2]\mF^{-1} = [|\xi|,-\partial_{\xi}^2],
\end{align*}
and the quadratic form associated to $h_0$ is invariant under the Fourier transform. Using these bounds the total commutator bound follows directly. Let $\Psi = (\phi,\alpha)^{\top}\in \mS(\R;\C^2)$. Then
\begin{align*}
     |\la \Psi,[H_0,N]\Psi\ra | &= |\la \phi, [|D|,x^2]\phi\ra +g\la \phi, [|x|,-\partial_x^2]\alpha\ra + g\la \alpha,[|x|,-\partial_x^2]\phi\ra | \\
     &\leq C(q(\phi)+q(\alpha)) = C\la \Psi,N\Psi\ra .
\end{align*}
Hence $H_0$ is essentially self-adjoint on $\mS(\R;\C^2)$.
\end{proof}

\section{Local Oscillator}\label{app:local}

In this section we consider the local analog of the eigenvalue problem \eqref{eq:spectral_problem} where the fractional Laplacian $|D|$ is replaced by the classical one-dimensional Laplacian $D^2= -\partial_x^2$. Specifically, we wish to find solutions to $H_{\rm loc}\Psi = \omega\Psi$ where $\Psi\in L^2(\R;\C^2)$ and the Hamiltonian $H_{\rm loc}$ is given by
\begin{align}\label{eq:local_Hamiltonian}
    H_{\rm loc} = \begin{pmatrix}
        -\partial_x^2 & g|x|\\ g|x| & \Omega
    \end{pmatrix}.
\end{align}
Physically, this corresponds to studying a scalar field with a quadratic dispersion relation $\omega(k)=k^2$ as observed in exciton-polariton condensates \cite{Amo_2009,Pinsker_2017,Wouters_2007,Ciuti_2005}.

\begin{theorem}\label{thm:spec_local}
Consider the Hamiltonian $H_{\rm loc}$ given in \eqref{eq:local_Hamiltonian}.
The collection of $L^2(\R;\C^2)$-eigenvalues of $H_{\rm loc}$ comprises an unbounded increasing sequence $\{\omega_n\}$ for which $\omega_n > \Omega$ for all $n\geq 0$. Moreover, we have the asymptotic formula as $n\to\infty$.
\begin{align*}
    \omega_n = (g(2n+1))^{2/3} + \frac{\Omega}{3} + O(n^{-2/3}).
\end{align*}
\end{theorem}

\begin{proof}
The Hamiltonian has a self-adjoint realization using an even simpler version of the Faris-Lavine argument used in Proposition \ref{prop:self_adjoint} (as the commutator $[-\partial_x^2,x^2]=-4x\partial_x-2$). Using the spectral-reduction argument of \Cref{lem:spectral_criterion}, $\omega\neq \Omega$ is an $L^2(\R;\C^2)$-eigenvalue of $H_{\rm loc}$ precisely when it is an $L^2(\R)$-eigenvalue of $L_{\omega}^{\rm loc}$
\begin{align*}
   L_{\omega}^{\rm loc} = -\partial_x^2 + \frac{g^2x^2}{\omega-\Omega}.
\end{align*}
A nearly identical argument to the one in the main body of the text shows that $\omega\neq\Omega$. Additionally, if $\omega<\Omega$, $L_{\omega}^{\rm loc}$ is an inverted oscillator and has no eigenvalues \cite{Barton_87,krason2020,Finster_2017}. Finally, suppose $\omega > \Omega$. In this case, the operator above is the Hamiltonian associated to a quantum harmonic oscillator, whose spectrum is discrete and well known.
\begin{align*}
    \lambda_n(\omega) = \frac{g(2n+1)}{\sqrt{\omega-\Omega}}.
\end{align*}
Thus, $\omega$ is an eigenvalue of $H$ exactly when
\begin{align*}
    \omega = \frac{g(2n+1)}{\sqrt{\omega-\Omega}},
\end{align*}
for some $n\geq 0$. For each such $n$ there is exactly one solution given by
\begin{align*}
    \omega_n\sqrt{\omega_n-\Omega} = g(2n+1).
\end{align*}
Defining $b_n = \sqrt{\omega_n-\Omega}$, then 
\begin{align*}
    b_n^3+\Omega b_n = g(2n+1),\quad \omega_n = \Omega + b_n^2.
\end{align*}
As $n\to\infty$, we have the asymptotics 
\begin{align*}
    b_n = (g(2n+1))^{1/3} -\frac{\Omega}{3}(g(2n+1))^{-1/3}+O(n^{-1})\,
\end{align*}
from which we obtain
\begin{align*}
     \omega_n = \Omega + b_n^2 = (g(2n+1))^{2/3} + \frac{\Omega}{3} + O(n^{-2/3}),\quad n\to\infty.
\end{align*}
\end{proof}

\printbibliography
\end{document}